\documentclass[10pt,journal]{IEEEtran}

\usepackage{paralist}
\usepackage{hyperref}

\usepackage{enumerate}
\usepackage{url}
\usepackage{amsthm}

\usepackage[pdftex]{graphicx}
\usepackage{algpseudocode}
\usepackage{algorithm}
\usepackage{amsmath}
\usepackage{amssymb}

\usepackage{upref}
\usepackage{amsfonts}
\usepackage{graphicx}
\usepackage{epic}
\usepackage{psfrag}
\usepackage{verbatim}
\usepackage{amsfonts}
\usepackage{amssymb}
\usepackage{array}
\usepackage{color}
\usepackage{textcomp}

\usepackage{tikz}
\usepackage{tikz,ifthen,calc}
\usetikzlibrary{trees,snakes}
\usetikzlibrary{arrows,shapes,positioning}

\newcommand{\eqdef}{\triangleq}

\newcommand{\cE}{{\cal E}}

\newcommand{\cG}{{\cal G}}
\newcommand{\cL}{{\cal L}}

\newcommand{\cS}{{\cal S}}
\newcommand{\cV}{{\cal V}}
\newcommand{\cP}{{\cal P}}

\newcommand{\cT}{{\cal T}}

\newcommand{\mK}{\mathcal{K}}

\newcommand{\sP}{\cP}
\newcommand{\sG}{\cG}

\newcommand{\Gr}{\smash{{\sG\kern-1.5pt}_q\kern-0.5pt(n,k)}}
\newcommand{\Grr}{\smash{{\sG\kern-1.5pt}_q\kern-0.5pt(n,r)}}
\newcommand{\Gfourk}{\smash{{\sG\kern-1.5pt}_q\kern-0.5pt(4k,2k)}}
\newcommand{\Gk}{\smash{{\sG\kern-1.5pt}_q\kern-0.5pt(n,k_1)}}
\newcommand{\Gkk}{\smash{{\sG\kern-1.5pt}_q\kern-0.5pt(n,k_2)}}
\newcommand{\Grtwo}{\smash{{\sG\kern-1.5pt}_2\kern-0.5pt(n,k)}}
\newcommand{\Gkone}{\smash{{\sG\kern-1.5pt}_q\kern-0.5pt(n,k_1)}}
\newcommand{\Gktwo}{\smash{{\sG\kern-1.5pt}_q\kern-0.5pt(n,k_2)}}
\newcommand{\Ps}{\smash{{\sP\kern-2.0pt}_q\kern-0.5pt(n)}}

\newtheorem{theorem}{Theorem}

\newtheorem{cor}{Corollary}
\newtheorem{lemma}[theorem]{Lemma}

\newtheorem{conjecture}{Conjecture}
\newtheorem{definition}[theorem]{Definition}

\newtheorem{example}{Example}

\begin{document}
\title{Constructions of Snake-in-the-Box Codes\\ for Rank Modulation}
\author{Michal Horovitz and Tuvi Etzion, Fellow, IEEE%
\thanks{Michal Horovitz is with the Department of Computer Science,
Technion --- Israel Institute of Technology, Haifa 32000, Israel.
(email: michalho@cs.technion.ac.il). This work is part of her Ph.D.
thesis performed at the Technion.
Tuvi Etzion is with the Computer Science Department, Technion--Israel
Institute of Technology, Haifa 32000, Israel
(e-mail: etzion@cs.technion.ac.il).
}
\thanks{This work was supported in part by the United States --- Israel
Binational Science~Foundation (BSF), Jerusalem, Israel, under Grant 2012016.
This paper was presented in part in the Information Theory and Applications workshop, San Diego, February 2014.}
}

\maketitle
\vspace{-0.5cm}
\begin{abstract}
Snake-in-the-box code is a Gray code which is capable of detecting
a single error. Gray codes are important in the context of
the rank modulation scheme which was suggested recently for
representing information in flash memories. For a Gray code
in this scheme the codewords are permutations, two consecutive
codewords are obtained by using the "push-to-the-top" operation,
and the distance measure is defined on permutations. In this paper
the Kendall's $\tau$-metric is used as the distance
measure. We present a general method
for constructing such Gray codes. We apply the method recursively
to obtain a snake of length $M_{2n+1}=((2n+1)(2n)-1)M_{2n-1}$ for permutations of $S_{2n+1}$,
from a snake of length $M_{2n-1}$ for permutations of~$S_{2n-1}$.
Thus, we have
$\lim\limits_{n\to \infty} \frac{M_{2n+1}}{S_{2n+1}}\approx 0.4338$,
improving on the previous known ratio of
$\lim\limits_{n\to \infty} \frac{1}{\sqrt{\pi n}}$.
By using the general method we also present a direct construction.
This direct construction is based on necklaces and it might
yield snakes of length $\frac{(2n+1)!}{2} -2n+1$ for permutations
of $S_{2n+1}$. The direct construction was applied successfully for
$S_7$ and $S_9$, and hence
$\lim\limits_{n\to \infty} \frac{M_{2n+1}}{S_{2n+1}}\approx 0.4743$.
\end{abstract}


\maketitle
\section{Introduction}

Flash memory is a non-volatile technology that is both electrically
programmable and electrically erasable. It incorporates a set of cells
maintained at a set of levels of charge to encode information.
While raising the charge level of a cell is an easy operation,
reducing the charge level requires the erasure of the whole block to which the cell belongs.
For this reason charge is injected into the cell over several iterations.
Such programming is slow and can cause errors since cells may be injected with extra unwanted charge.
Other common errors in flash memory cells are due to charge leakage and reading disturbance that
may cause charge to move from one cell to its adjacent cells.
In order to overcome these problems, the novel framework
of \emph{rank modulation} was introduced in~\cite{JMSB}.
In this setup the information is carried by the relative ranking of the
cells' charge levels and not by the absolute values of the charge levels.
This allows for more efficient programming of cells, and coding by the ranking of the cells' levels
is more robust to charge leakage than coding by their actual values.
In this model codes are subsets of~$S_n$,
the set of all permutations on $n$ elements,
and the codewords are members of $S_n$,
where each permutation corresponds to a ranking of $n$ cells' levels from the highest one to the lowest.
For example, the charge levels $(c_1,c_2,c_3,c_4)=(5,1,3,4)$ are represented by the codeword $[1,4,3,2]$
since the first cell has the highest level, the forth cell has the next highest level and so on.

To detect and/or correct errors caused by injection of extra charge or due
to charge leakage we will use an appropriate
distance measure.
Several metrics on permutations are used for this purpose. In this
paper we will consider only the Kendall's $\tau$-metric~\cite{JSB10,KeGi90}.
The Kendall's $\tau$-distance between two permutation $\pi_1$ and $\pi_2$ in $S_n$
is the minimum adjacent transpositions required to obtained $\pi_2$ from $\pi_1$, where adjacent transposition is an exchange of two distinct adjacent elements.
For example, the Kendall's $\tau$-distance between $\pi_1=[2,1,4,3]$ and $\pi_2=[2,4,3,1]$ is $2$ as
$[2,1,4,3]\to [2,4,1,3] \to [2,4,3,1]$.
Two permutations in this metric are at distance
one if they differ in exactly one pair of adjacent elements.
Distance one between these two permutations represent an exchange of two
cells, which are adjacent in the permutation, due to a small change in
their charge level which change their order.

Gray codes are very important in the context of rank modulation
as was explained in~\cite{JMSB}. They are used in many other applications,
e.g.~\cite{EtPa96,SaWi95}. An excellent survey on Gray codes is given in~\cite{Sav97}.
The usage of Gray codes for
rank modulation was also discussed in~\cite{GLSB11,GLSB13,JMSB,YeSc12}.
The permutations of $S_n$ in the rank modulation scheme
represent "new" logical levels of the flash memory. The codewords in
the Gray code provide the order of these levels which should be implemented
in various algorithms with the rank modulation scheme.
Usually, a Gray code is just a simple cycle in a graph, in which the edges are defined between
vertices with distance one in a given metric. Two adjacent
vertices in the graph represent on one hand two elements
whose distance is one by the given metric; and on
the other hand a move from a vertex to a vertex implied
by an operation defined by the metric. A snake-in-the-box code
is a Gray code in which two elements in the code are not
adjacent in the graph, unless they are consecutive in the
code. Such a Gray code can detect a single error in a codeword.
Snake-in-the-box codes were mainly discussed in the context of
the Hamming scheme, e.g.~\cite{AbKa88}.

In the rank modulation scheme the Gray code is defined slightly
different since the operation is not defined by a metric.
The permutation is defined by the
order of the charge levels, from the highest one to the lowest one.
From a given ranking of the charge levels, which defines a permutation,
the next ranking is obtained by raising
the charge level of one of the cells to be the highest level.
This operation, called "push-to-the-top", is used in the rank modulation scheme.
For example, the charge levels $(c_1,c_2,c_3,c_4)=(5,1,3,4)$ are represented by the codeword $[1,4,3,2]$,
and by applying push-to-the-top operation on the second cell which has the lowest charge level,
we have, for example, the charge levels $(c_1,c_2,c_3,c_4)=(5,6,3,4)$ which are represented by the codeword $[2,1,4,3]$.
Hence, the permutation $\pi_2$ can follow the permutation $\pi_1$
if $\pi_2$ is obtained from $\pi_1$ by applying a push-to-the-top operation on $\pi_1$.
Therefore, the related graph is directed with an outgoing edge from the vertex
which represents $\pi_1$ into the vertex which represents $\pi_2$.
On the other hand, one possible metric for the scheme is
the Kendall's $\tau$-metric.
A Gray code
(and a snake-in-the-box code as a special case)
related to the rank modulation scheme is a directed simple cycle in the graph.
In a snake-in-the-box code, related to this scheme, there is another requirement that
the Kendall's $\tau$-distance between any two codewords
is at least two, including consecutive codewords.
For example,
$C=([1,2,3,4],\ [4,1,2,3],\ [2,4,1,3],\ [3,2,4,1],\ [4,3,2,1], \newline [1,4,3,2],\ [3,1,4,2],\ [2,3,1,4])$
is a snake-in-the-box code in $S_4$
obtained by applying a push-to-the-top operation on the lowest cell at each time.
The Kendall's $\tau$-distance between any two permutations in $C$ is at least~$2$.

One of the most important problems in the research on
snake-in-the-box codes is to construct the
largest possible code for the given graph. In a snake-in-the-box
code for the rank modulation scheme
we would like to find such a code with the largest number of permutations.
In a recent paper by Yehezkeally and Schwartz~\cite{YeSc12}, the authors
constructed a snake-in-the-box code of length
${M_{2n+1}=(2n+1)(2n-1) M_{2n-1}}$ for permutations of $S_{2n+1}$,
from a snake of length $M_{2n-1}$ for permutations of~$S_{2n-1}$.
We will improve on this result by constructing a snake of length
$M_{2n+1}=((2n+1)2n-1) M_{2n-1}$ for permutations of $S_{2n+1}$,
from a snake of length $M_{2n-1}$ for permutations of~$S_{2n-1}$.
Thus, we have
$\lim\limits_{n\to \infty} \frac{M_{2n+1}}{S_{2n+1}}\approx 0.4338$,
improving on the previous known ratio of
$\lim\limits_{n\to \infty} \frac{1}{\sqrt{\pi n}}$ \cite{YeSc12}.
For these constructions of snake-in-the-box codes we
need an initial snake-in-the-box code and the largest one known
to start both constructions is a snake of length 57 for
permutations of $S_5$. We also propose a direct construction
to form a snake of length
$\frac{(2n+1)!}{2} -2n+1$ for permutations
of $S_{2n+1}$. The direct construction was applied successfully for
$S_7$ and $S_9$. This implies better initial condition for the
recursive constructions, and the ratio $\lim\limits_{n\to \infty} \frac{M_{2n+1}}{S_{2n+1}}\approx 0.4743$.

The rest of this paper is organized as follows.
In Section~\ref{sec:pre} we will define the basic concepts
of Gray codes in the rank modulation
scheme, the push-to-the-top operation,
and the Kendall's $\tau$-metric required in this paper.
In Section~\ref{sec:main} we present the main ideas
and a framework for constructions of snake-in-the-box codes.
In Section~\ref{sec:recursive} we present a recursive construction
based on the given framework. This construction is used to obtain
snake-in-the-box codes longer than the ones known before.
In~Section~\ref{sec:direct}, based on the framework,
we present an idea for a
direct construction based on necklaces.
The construction is used to obtain snake-in-the-box codes
of length $\frac{(2n+1)!}{2} -2n+1$ in $S_{2n+1}$,
which we believe are optimal. The construction was
applied successfully on $S_7$ and on $S_9$,
and we conjecture that it can be applied
on $S_n$ for any odd $n > 6$.
Conclusions and problems for future research are
presented in Section~\ref{sec:conclude}.

\section{Preliminaries}
\label{sec:pre}

In this section we will repeat some notations defined and mentioned
in~\cite{YeSc12}, and we also present some other definitions.

Let $[n]\triangleq\{1, 2, \ldots , n\}$
and let $\pi= [a_1, a_2, \ldots , a_n]$ be a permutation
over $[n]$, i.e., a permutation in $S_n$, such that for each $i\in [n]$ we have that $\pi(i)=a_i$.

Given a set $\cS$ and a subset of transformations
$T\subseteq \{f|f:\cS\to \cS\}$, a \emph{Gray code}
over $\cS$ of size $M$, using transitions from
$T$, is a sequence $C=(c_0, c_1, \ldots , c_{M-1})$ of
$M$ distinct elements from $\cS$, called \emph{codewords}, such that for each
$j\in[M-1]$ there exists a $t\in T$ for which $c_j=t(c_{j-1})$.
The Gray code is called \emph{complete} if $M=|\cS|$,
and \emph{cyclic} if there exists $t\in T$ such that $c_0=t(c_{M-1})$.
Throughout this paper we will consider only cyclic Gray codes.

In the context of rank modulation for flash memories,
$\cS=S_n$ and the set of transformations $T$ comprises of push-to-the-top
operations. We denote by $t_i$ the \emph{push-to-the-top}
operation on index $i$, $2 \leq i \leq n$, defined by
\begin{align*}
t_i(&[a_1,\ldots,a_{i-1},a_i,a_{i+1},\ldots,a_n])=\\
&[a_i,a_1,\ldots,a_{i-1},a_{i+1},\ldots,a_n].
\end{align*}
and a \emph{p-transition} will be an abbreviated notation for a push-to-the-top operation.

A sequence of p-transitions will be called a \emph{transitions sequence}.
A permutation $\pi_0$ and a transitions sequence
$t_1 , t_2, \ldots t_\ell$ define a sequence of permutations
$\pi_0 , \pi_1 , \pi_2 ,\ldots , \pi_{\ell-1} ,\pi_\ell$,
where $\pi_i = t_i (\pi_{i-1})$, for each $i$, $1 \leq i \leq \ell$. This
sequence is a cyclic Gray code, if $\pi_\ell = \pi_0$
and for each $0 \leq i < j < \ell$, $\pi_i \neq \pi_j$.
In the sequel the word cyclic will be omitted

Given a permutation $\pi=[a_1,a_2,\ldots,a_n] \in S_n$,
an \emph{adjacent transposition} is an exchange of two distinct adjacent elements
$a_i,a_{i+1}$, in $\pi$, for some $1\leq i\leq n-1$.
The result of such an adjacent transposition is the permutation
$[a_1,\ldots,a_{i-1},a_{i+1},a_i,a_{i+2},\ldots,a_n]$.
The \emph{Kendall's} $\tau$-\emph{distance}~\cite{KeGi90} between two permutations
$\pi_1,\pi_2 \in S_n$ denoted by $d_K (\pi_1,\pi_2)$
is the minimum number of adjacent transpositions
required to obtain the permutation $\pi_2$ from the permutation $\pi_1$.
A \emph{snake-in-the-box code} is a Gray code in which for each two
permutations $\pi_1$ and $\pi_2$ in the code we have
$d_K(\pi_1,\pi_2) \geq 2$.
Hence, a snake-in-the-box code is a Gray code capable of detecting
one Kendall's $\tau$-error. We will call such a snake-in-the-box
code a $\mK$\emph{-snake}.
We further denote by $(n,M,\mK)$\emph{-snake} a $\mK$-snake of size $M$
with permutations from $S_n$. A $\mK$-snake can be
represented in two different equivalent ways:
\begin{itemize}
{\setlength\itemindent{-9pt}
\item the sequence of codewords (permutations),
\item the transitions sequence along with the first permutation.}
\end{itemize}

Let $\cT$ be a transitions sequence and let $\pi$ be a permutation in $S_n$.
If a $\mK$-snake is obtained by applying $\cT$ on $\pi$
then a $\mK$-snake will be obtained by using any other
permutation from $S_n$ instead of $\pi$. This is a simple observation
from the fact that $t( \pi_2 ( \pi_1 )) = \pi_2 (t(\pi_1))$,
where $t$ is a p-transition and $\pi_2 (\pi_1)$ refers to applying the
permutation $\pi_2 \in S_n$ on the permutation $\pi_1 \in S_n$.
In other words applying $\cT$ on a different permutation
just permute the symbols, by a fixed given permutation,
in all the resulting permutations when $\cT$ is applied on $\pi$. Therefore,
such a transitions sequence~$\cT$ will be called an \emph{S-skeleton}.

For a transitions sequence ${\sigma=t_{k_1}, t_{k_2}, \ldots t_{k_{\ell}}}$
and a permutation $\pi \in S_n$,
we denote by $\sigma\left(\pi\right)$, the permutation obtained by applying
the sequence of p-transitions in $\sigma$ on $\pi$,
i.e., $t_{k_1}$ is applied on $\pi$, $t_{k_2}$ is applied on $t_{k_1} (\pi)$, and so on.
In other words, $\sigma\left(\pi\right)= (t_{k_1}\circ t_{k_2}\circ \ldots \circ t_{k_{\ell}})(\pi)=
t_{k_{\ell}}\left(t_{k_{l-1}}\left(\ldots t_{k_2}\left(t_{k_1}\left(\pi\right)\right)\right)\right)$.
Let $\sigma_1,\sigma_2$ be two transitions sequences.
We say that $\sigma_1$ and  $\sigma_2$ are \emph{matching sequences},
and denote it by $\sigma_1 \leftrightsquigarrow \sigma_2$,
if for each ${\pi\in S_n}$ we have $\sigma_1(\pi)=\sigma_2(\pi)$.

In~\cite{YeSc12} it was proved that a Gray code with permutations from $S_n$ using only
p-transitions on odd indices is a $\mK$-snake. By starting with an even permutation and using only
p-transitions on odd indices we get a sequence
of even permutations, i.e., a subset of the $A_n$, the alternating group of order $n$.
This observation saves us the need to check whether a Gray code is
in fact a $\mK$-snake, at the cost of restricting the
permutations in the $\mK$-snake to the set of even permutations.
However, the following assertions were also proved in~\cite{YeSc12}.
\begin{itemize}
\item If $C$ is an $(n,M,\mK)$-snake then $M\leq \frac{{|S_n|}}{2}$.
\item If $C$ is an $(n,M,\mK)$-snake which contains a p-transition
on an even index then ${M\leq{\frac{{|S_n|}}{2}-\frac{1}{n-1}\binom{\lfloor n/2 \rfloor -1}{2}}}$.
\end{itemize}
This motivates not to use p-transitions on even indices.
Since we will use only p-transitions on odd indices,
we will describe our constructions only for even permutations with odd length.

\section{Framework for Constructions of $\mK$-Snakes}
\label{sec:main}

In this section we present a framework for constructing $\mK$-snakes in $S_{2n+1}$.
Our snakes will contain only even permutations. We start by partitioning
the set of even permutations of $S_{2n+1}$ into classes. Next, we describe
how to merge $\mK$-snakes of different classes into one $\mK$-snake.
We conclude this section by describing how to combine most of these classes
by using a hypergraph whose vertices represent the classes and whose
edges represent the classes that can be merge together in one step.

We present two constructions for a $(2n+1,M_{2n+1},\mK)$-snake,
$C_{2n+1}$, one recursive and one direct. In this section
we present the framework for these constructions.
First, the permutations of $A_{2n+1}$,
the set of even permutations from~$S_{2n+1}$, are partitioned into classes,
where each class induces one $\mK$-snake which contains permutations only from the class.
All these snakes have the same S-skeleton.
Let $L_{2n+1}$ be the set of all the classes.

The construction of $C_{2n+1}$ from the $\mK$-snakes of $L_{2n+1}$
proceeds by a sequence of joins, where at each step
we have a main $\mK$-snake, and
two $\mK$-snakes from the remaining $\mK$-snakes
of $L_{2n+1}$ are joined to the current main $\mK$-snake.
A join is performed by replacing one transition
in the main $\mK$-snake with a matching sequence.

In order to join the $\mK$-snakes we need the following lemmas,
for which the first can be easily verified. In the sequel, let
${\sigma^{k}\triangleq\underbrace{\sigma\circ \sigma \circ \ldots \circ \sigma}_{k \, times}}$,
i.e., performing the transitions sequence $\sigma$, $k$ times.
\begin{lemma}
\label{lem:helpLem1}
If $\alpha,\beta \in S_n$ then $\beta{=}t_{i}(\alpha)$ if and only if $\alpha{=}t_i^{i-1}(\beta)$.
\end{lemma}
\begin{lemma}
\label{lem:mainLem}
If ${i \in [n-2]}$ then $t_i \leftrightsquigarrow t_{i+2}\circ (t_i^{i-1} \circ t_{i+2})^{2}$.
\end{lemma}
\begin{proof}
Let ${\alpha=[a_1, a_2, \ldots ,a_i, a_{i+1}, a_{i+2}, \ldots , a_n]}$ be a permutation over [n].

\begin{align*}
\begin{tabular}{lllllllll}
$t_{i+2}(\alpha)$  & \\
\hspace*{\parindent} $=[a_{i+2}, a_1, \ldots , a_i, a_{i+1}, a_{i+3},$  &  $\ldots,a_n]$,\tabularnewline
$t_i^{i-1}(t_{i+2}(\alpha))$ & \\
\hspace*{\parindent} $=[a_1, a_2, \ldots , a_{i-1}, a_{i+2}, a_i, a_{i+1}, a_{i+3},$  &  $\ldots,a_n]$,\tabularnewline
$t_{i+2}(t_i^{i-1}(t_{i+2}(\alpha)))$  & \\
\hspace*{\parindent} $=[a_{i+1}, a_1, a_2, \ldots , a_{i-1},  a_{i+2}, a_i, a_{i+3},$  &  $\ldots,a_n]$,\tabularnewline
$t_i^{i-1}(t_{i+2}(t_i^{i-1}(t_{i+2}(\alpha))))$  & \\
\hspace*{\parindent} $=[a_1, a_2, \ldots ,a_{i-1}, a_{i+1}, a_{i+2}, a_i, a_{i+3},$  &  $\ldots,a_n]$,\tabularnewline
\intertext{and hence we have,}
$t_{i+2}(t_i^{i-1}(t_{i+2}(t_i^{i-1}(t_{i+2}(\alpha)))))$ & \\
\hspace*{\parindent} $=[a_i, a_1, \ldots , a_{i-1}, a_{i+1}, a_{i+2},$  &  $\ldots,a_n]$ \tabularnewline
\hspace*{\parindent} $=t_i(\alpha)$.
\end{tabular}
\end{align*}
\end{proof}
\vspace{-4pt}
\begin{cor}
\label{cor:mainLem}
If $\pi \in S_{2n+1}$ then $t_{2n-1}(\pi)=t_{2n+1}\left(t_{2n-1}^{2n-2}\left(t_{2n+1}\left(t_{2n-1}^{2n-2}\left(t_{2n+1}(\pi)\right)\right)\right)\right)$.
\end{cor}

Lemma~\ref{lem:mainLem} can be generalized as follows
(the following lemma is given
for completeness, but it will not be used in the sequel, and hence
its proof is omitted).

\begin{lemma}
\label{lem:mainLemExp}
If ${i,j\in [n]}$ and ${|i-j|{=}k}$, then ${t_i \leftrightsquigarrow t_j\circ (t_i^{i-1} \circ t_j)^{k}}$.
\end{lemma}

The partition of $A_{2n+1}$ into the set of classes $L_{2n+1}$
should satisfy the following properties:
\begin{itemize}
\item[(P1)] The last two ordered elements of two permutations in same class are equal.
\item[(P2)] Any two permutations which differ only by a cyclic shift of the first $2n-1$ elements, belong to the same class.
\end{itemize}

\begin{cor}
Let $\pi$ be a permutation in $A_{2n+1}$.
\begin{itemize}
\item $\pi$ and $t_{2n+1}(\pi)$ belong to different classes in $L_{2n+1}$.
\item $\pi$ and $t_{2n-1}(\pi)$ belong to the same class in $L_{2n+1}$.
\end{itemize}
\end{cor}

We continue now with the description of the method to join the
$\mK$-snakes of $L_{2n+1}$ into $C_{2n+1}$.
In the rest of the paper, $A_{2n+1}$ is partitioned into classes according to the last two ordered elements in the permutations.
Let $[x,y]$ denote the class of $A_{2n+1}$ in which the last ordered pair in the permutations is $(x,y)$.
Let $\cT$ be the S-skeleton of the $\mK$-snakes in $L_{2n+1}$.
Let $C_{\cT}^{\pi}$ be a $\mK$-snake for which
$\cT$ is its transitions sequence, and $\pi$ is its first permutation.
If $\pi$ belongs to the class $[x,y]$,
we say that $C_{\cT}^{\pi}$ represents the class $[x,y]$.
Note that all the permutations in $C_{\cT}^{\pi}$ belong to the same class.

The transitions sequence $\cT$
should satisfy the following properties (these properties are needed in order to make the required joins of cycles):
\begin{itemize}
\item[(P3)] $t_{2n-1}$ is the last transition in $\cT$.
\item[(P4)] Given a permutation $\pi=[a_1,\ldots, a_{2n},a_{2n+1}]$,
for each $x\in [2n+1]\setminus \{a_{2n},a_{2n+1}\}$
there exists a permutation $\pi' \in C_{\cT}^{\pi}$
whose last ordered three elements are $(x,a_{2n},a_{2n+1})$.
\end{itemize}

\begin{cor}
\label{cor:mergeCyclesCondition}
For each class $[x,y]$, a permutation $\pi \in [x,y]$, and $z\in [2n+1]\setminus \{x,y\}$,
there exists a permutation $\pi' \in C_{\cT}^{\pi}$ whose last ordered three elements are $(z,x,y)$,
followed by the permutation $t_{2n-1}(\pi')$.
\end{cor}

\begin{lemma}
\label{prop:mergeClasses1}
Let $C$ be a $\mK$-snake which doesn't contain any permutation from the classes $[y,z]$ or $[z,x]$, let
$\pi=[a_1, a_2, \ldots, a_{2n-2},z, \mathbf{x,y}]$
be a permutation in $C$ followed by $t_{2n-1}$, and let
$\sigma$ be a transitions sequence such that $\cT = \sigma \circ t_{2n-1}$.
Then replacing this $t_{2n-1}$ transition in $C$, with
\begin{equation*}
t_{2n+1}\circ \sigma \circ t_{2n+1}\circ \sigma \circ t_{2n+1} ,
\end{equation*}
joins two $\mK$-snakes representing the classes ${[y, z]}$ and ${[z, x]}$ into $C$ (after $\pi$).
\end{lemma}
\begin{proof}
Observe that by Lemma \ref{lem:helpLem1} we have~${\sigma \leftrightsquigarrow t_{2n-1}^{2n-2}}$.
Thus, we have
\begin{align*}
\pi=&
    \left.\begin{array}{l}
    		\left[ a_1, a_2, \ldots,  a_{2n-2}, z, \mathbf{{x},y} \right]
    \\ \end{array}\right. \\
    &
    \left.\begin{array}{l}
    		\downarrow {t_{2n+1}}
    \\ \end{array}\right. \\
    &
    \left.\begin{array}{l}
    		\left[ y, a_1, a_2, \ldots, a_{2n-2}, \mathbf{z, x} \right]   \\
    		\downarrow {\sigma \leftrightsquigarrow t_{2n-1}^{2n-2}} \\
    		\left[ a_1, a_2, \ldots, a_{2n-2}, y, \mathbf{z, x} \right]
    \\ \end{array}\right\}
    \left.\begin{array}{l}
    \mK-snake\\
    for\ \left[z,x\right] \\
    \end{array}\right. \\
    &
    \left.\begin{array}{l}
    		\downarrow {t_{2n+1}}
    \\ \end{array}\right. \\
    &
    \left.\begin{array}{l}
    		\left[ x, a_1, a_2, \ldots, a_{2n-2}, \mathbf{y, z}\right] \\
    		\downarrow {\sigma \leftrightsquigarrow t_{2n-1}^{2n-2}} \\
    		\left[ a_1, a_2, \ldots, a_{2n-2}, x,  \mathbf{y, z}\right]
    \\ \end{array} \right\}
    \left.\begin{array}{l}
    \mK-snake\\
    for\ \left[y, z\right] \\
    \end{array}\right. \\
    &
    \left.\begin{array}{l}
    		\downarrow {t_{2n+1}}
    \\ \end{array}\right. \hphantom{\left[ a_1, a_2, \ldots, x,\mathbf{y, z} \right]}
     \left.\begin{array}{l}
        return\ to\ the\ \\
        \mK-snake\ C  \\
        \end{array}\right. \\
t_{2n-1}(\pi)=&
    \left.\begin{array}{l}
    \left[ z, a_1, a_2, \ldots, a_{2n-2}, \mathbf{x, y}\right]
    \end{array}\right.
\end{align*}
\end{proof}
The next step is to present an order for merging all the $\mK$-snakes of $L_{2n+1}$,
except one, into $C_{2n+1}$.
This step will be performed by
translating the merging problem into a 3-graph problem.
We start with a sequence of definitions taken from~\cite{hypergraph}.
\begin{definition}
\label{def:3-graph}
A 3-graph (also called a 3-uniform hypergraph) $H=(V,E)$ is a hypergraph where $V$ is a set of vertices
and $E\subseteq \binom{V}{3}$. A~hyperedge of $H$ will be called triple.\\
A path in $H$ is an alternating sequence of $\ell+1$ distinct
vertices and $\ell$ distinct triples:
$v_0, e_1, v_1, \ldots ,v_{\ell-1}, e_{\ell}, v_{\ell}$,
with the property that $\forall i\in [\ell]: v_{i-1},v_{i} \in e_{i}$.\\
A cycle is a closed path, i.e. $v_0 = v_{\ell}$.\\
A sub-3-graph contains a subset $E'\subseteq E$ and the subset
$V'\subseteq V$ which contains all the vertices in $E'$.\\
A tree $T$ in $H$ is a connected sub-3-graph of $H$ with no cycles.
\end{definition}

\pagebreak
Let $H_{2n+1}=(V_{2n+1},E_{2n+1})$ be a 3-graph defined as follows:
\begin{align*}
V_{2n+1}&=\{[x,y]~:~x,y \in [2n+1], x\ne y\},\\
E_{2n+1}&=\{\{[x,y],[y,z],[z,x]\}~:\\
&\hphantom{=\{\}} x,y,z\in [2n+1], x\ne y, x\ne z, y\ne z\}.
\end{align*}
We denote a hyperedge $\{[x,y],[y,z],[z,x]\}$, where $x<y$ and $x<z$, by the triple $\langle x,y,z \rangle$.

The vertices in $H_{2n+1}$ correspond to the classes in the set $L_{2n+1}$.
Each $e\in E_{2n+1}$ contains three vertices,
which correspond to three classes.
These three classes can be represented
by three $\mK$-snakes, generated from the S-skeleton,
which can be merged together
by Corollary~\ref{cor:mergeCyclesCondition} and Lemma~\ref{prop:mergeClasses1}.
Note that for any two edges $e_1 , e_2$ in $H_{2n+1}$
either $e_1 \cap e_2 = \varnothing$ or $| e_1 \cap e_2 | =1$.
Let $T_{2n+1}=(V_{T_{2n+1}},E_{T_{2n+1}})$ be a tree in $H_{2n+1}$.
We join $|V_{T_{2n+1}}|$ $\mK$-snakes which represent $|V_{T_{2n+1}}|$ classes of $L_{2n+1}$
to form the $\mK$-snake $C_{2n+1}$,
by Corollary~\ref{cor:mergeCyclesCondition} and Lemma~\ref{prop:mergeClasses1}.
The hyperedges which represent the joins which are performed are determined by $T_{2n+1}$,
but these joins are not unique,
and hence they can yield different final $\mK$-snakes.
The order in which the hyperedges are selected for these joins is also not unique,
but this order doesn't affect the final $\mK$-snakes.
The size of the $\mK$-snake $C_{2n+1}$ depends on the number of vertices
in the tree $T_{2n+1}$.
A tree in a 3-graph
contains an odd number of vertices \cite{hypergraph}.
Since in $H_{2n+1}$ there are $(2n+1)(2n)$ vertices
it follows that there is no tree in $H_{2n+1}$ which contains
all the vertices of $V_{2n+1}$. This motivates the following definition.
\begin{definition}
A nearly spanning tree in a 3-graph $H=(V,E)$
is a tree in $H$
which contains all the vertices of $V$ except one.
\end{definition}

Now, let $T_{2n+1}$ be a nearly spanning tree in $H_{2n+1}$.

\begin{example}
\label{Exm:T5}
One choice for $T_{5}$ is given below.\\
The edges in the tree $T_{5}$ are:

\begin{tabular}{lll}
\label{treeFor5}
\hspace*{-\parindent}
\hspace*{-\parindent}
$\langle 1,2,5 \rangle,$  &  $\langle 1,2,4 \rangle,$  &  $\langle 1,2,3 \rangle,$ \tabularnewline
\hspace*{-\parindent}
\hspace*{-\parindent}
$\langle 1,4,5 \rangle,\langle 2,5,4 \rangle,$  &  $\langle 1,3,4 \rangle,\langle 2,4,3 \rangle,$  &  $\langle 1,5,3 \rangle, \langle 2,3,5 \rangle.$
\end{tabular}
The order of merging $\mK$-snakes
from these classes obtained by this choice of $T_5$
can be chosen as follows.
\begin{enumerate}[(1)]
{\setlength\itemindent{10pt}
\item vertex $\hphantom{ii}[1,2]$;
\item vertices $[3,1]$, $[2,3]$, (through the edge $\langle 1,2,3 \rangle$);
\item vertices $[4,1]$, $[2,4]$, (through the edge $\langle 1,2,4 \rangle$);
\item vertices $[5,1]$, $[2,5]$, (through the edge $\langle 1,2,5 \rangle$);
\item vertices $[5,3]$, $[1,5]$, (through the edge $\langle 1,5,3 \rangle$);
\item vertices $[5,2]$, $[3,5]$, (through the edge $\langle 2,3,5 \rangle$);
\item vertices $[3,4]$, $[1,3]$, (through the edge $\langle 1,3,4 \rangle$);
\item vertices $[3,2]$, $[4,3]$, (through the edge $\langle 2,4,3 \rangle$);
\item vertices $[4,5]$, $[1,4]$, (through the edge $\langle 1,4,5 \rangle$);
\item vertices $[4,2]$, $[5,4]$, (through the edge $\langle 2,5,4 \rangle$).}
\end{enumerate}

Using the S-skeleton $\cT=t_3,t_3,t_3$ of the $(3,3,\mK)$-snake,
the snake-in-the-box code which is obtained
by $T_{5}$ is a $(5,57,\mK)$-snake presented in Figure~\ref{fig:snake5}.
There is no $(5,M,\mK)$-snake for which $M>57$ \cite{YeSc12}.
The S-skeleton of this code is $\sigma ^ 3$, where
\begin{align*}
\sigma=	& t_5,t_5,t_3,t_3,t_5,t_3,t_3,t_5,t_3,\\
		& t_5,t_5,t_3,t_3,t_5,t_3,t_3,t_5,t_3,t_5
\end{align*}
\begin{figure*}[htbp]
\begin{center}
\small
\setlength{\tabcolsep}{1.5pt}
\begin{tabular}{*{56}{c|}c}
 3 & 2 & 1 & 3 & 2 & 5 & 3 & 2 & 4 & 3 & 1 & 5 & 3 & 1 & 2 & 3 & 1 & 4 & 3 & 5 &
 2 & 1 & 5 & 2 & 4 & 5 & 2 & 3 & 5 & 1 & 4 & 5 & 1 & 2 & 5 & 1 & 3 & 5 & 4 & 2 &
 1 & 4 & 2 & 3 & 4 & 2 & 5 & 4 & 1 & 3 & 4 & 1 & 2 & 4 & 1 & 5 & 4 \\
 4 & 3 & 2 & 1 & 3 & 2 & 5 & 3 & 2 & 4 & 3 & 1 & 5 & 3 & 1 & 2 & 3 & 1 & 4 & 3 &
 5 & 2 & 1 & 5 & 2 & 4 & 5 & 2 & 3 & 5 & 1 & 4 & 5 & 1 & 2 & 5 & 1 & 3 & 5 & 4 &
 2 & 1 & 4 & 2 & 3 & 4 & 2 & 5 & 4 & 1 & 3 & 4 & 1 & 2 & 4 & 1 & 5 \\
 5 & 4 & 3 & 2 & 1 & 3 & 2 & 5 & 3 & 2 & 4 & 3 & 1 & 5 & 3 & 1 & 2 & 3 & 1 & 4 &
 3 & 5 & 2 & 1 & 5 & 2 & 4 & 5 & 2 & 3 & 5 & 1 & 4 & 5 & 1 & 2 & 5 & 1 & 3 & 5 &
 4 & 2 & 1 & 4 & 2 & 3 & 4 & 2 & 5 & 4 & 1 & 3 & 4 & 1 & 2 & 4 & 1 \\
 1 & 5 & 4 & 4 & 4 & 1 & 1 & 1 & 5 & 5 & 2 & 4 & 4 & 4 & 5 & 5 & 5 & 2 & 2 & 1 &
 4 & 3 & 3 & 3 & 1 & 1 & 1 & 4 & 4 & 2 & 3 & 3 & 3 & 4 & 4 & 4 & 2 & 2 & 1 & 3 &
 5 & 5 & 5 & 1 & 1 & 1 & 3 & 3 & 2 & 5 & 5 & 5 & 3 & 3 & 3 & 2 & 2 \\
 2 & 1 & 5 & 5 & 5 & 4 & 4 & 4 & 1 & 1 & 5 & 2 & 2 & 2 & 4 & 4 & 4 & 5 & 5 & 2 &
 1 & 4 & 4 & 4 & 3 & 3 & 3 & 1 & 1 & 4 & 2 & 2 & 2 & 3 & 3 & 3 & 4 & 4 & 2 & 1 &
 3 & 3 & 3 & 5 & 5 & 5 & 1 & 1 & 3 & 2 & 2 & 2 & 5 & 5 & 5 & 3 & 3 \\
\end{tabular}
\end{center}
\caption{A $(5,57,\mK)$-snake obtained by $T_5$} \label{fig:snake5}
\end{figure*}
\end{example}

\begin{theorem}
\label{thm:treeCompose}
If $n\geq 2$, then there exists a nearly spanning tree $T_{2n+1}$ in $H_{2n+1}$ which doesn't include the vertex $[2,1]$.
\end{theorem}
\begin{proof}
We present a recursive construction for such a nearly spanning tree.
We start with the nearly spanning tree given in Example~\ref{Exm:T5}.
Note that $T_5$ doesn't include the vertex $[2,1]$.
Assume that there exists a nearly spanning tree,~$T_{2n-1}$,
in $H_{2n-1}$, which doesn't include the vertex $[2,1]$.
Note that $H_{2n-1}$ is a sub-graph of~$H_{2n+1}$
and therefore~$T_{2n-1}$ is a tree in $H_{2n+1}$.
The vertices of $H_{2n+1}$ which are not spanned by $T_{2n-1}$ are
\begin{itemize}
\item $[x,2n], [2n,x], [x,2n+1], [2n+1,x]$ for each $x\in [2n-1]$,
\item $[2n,2n+1], [2n+1,2n]$,
\item $[2,1]$.
\end{itemize}
The nearly spanning tree $T_{2n+1}$ is constructed from $T_{2n-1}$ as follows.
For each $x$, $2\leq x \leq 2n-2$, the edges
${\langle x,x+1,2n \rangle}$ and $\langle x,x+1,2n+1 \rangle$
are joined to $T_{2n+1}$;
also the edges
$\langle 1,2,2n \rangle$,
$\langle 1,2n,2n-1 \rangle$,
$\langle 1,2n+1,2n-1 \rangle$,
$\langle 1,2n,2n+1 \rangle$, and
$\langle 2,2n+1,2n \rangle$ are joined to $T_{2n+1}$.
It is easy to verify that all the vertices of $H_{2n+1}$
which are not spanned by $T_{2n-1}$ (except for $[2,1]$)
are contained in the list of the edges which are joined to $T_{2n-1}$.
When an edge is joined to the tree it has one vertex which is already in the
tree and two vertices which are not on the tree. Hence, connectivity is
preserved and no cycle is formed.
Hence, it is easy to verify that by joining these edges to $T_{2n-1}$
we form a nearly spanning tree in $H_{2n+1}$.
\end{proof}

\begin{example}
\label{Exm:T7}
By using Theorem~\ref{thm:treeCompose} and the
nearly spanning tree $T_5$ of Example~\ref{Exm:T5}
we obtain the spanning tree $T_7$ depicted in Figure~\ref{fig:T7}.
The dashed boxes edges and the double lines nodes are added to $T_5$ in order to form $T_7$.
\pagestyle{empty}
\tikzstyle{level 1}=[sibling angle=90]
\tikzstyle{level 2}=[sibling angle=80]
\tikzstyle{level 3}=[sibling angle=50]
\tikzstyle{level 4}=[sibling angle=50]
\tikzstyle{every node}=[draw]
\tikzstyle{edge from parent}=[draw]
\begin{figure*}[htbp]
\centering
\begin{tikzpicture}[grow cyclic,very thick,level distance=10mm,
                    cap=round, scale=0.9]
                    \tikzset{VertexStyle/.style = {shape          = circle,
                                                     draw,
                                                     text           = black,
                                                     inner sep      = 1pt,
                                                     outer sep      = 0pt,
                                                     minimum size   = 8 pt,
                                                     scale=0.9}}
                    \tikzset{Vertex7Style/.style = {shape          = circle,
                                                     draw,
                                                     thick,
                                                     double,
                                                     text           = black,
                                                     inner sep      = 1pt,
                                                     outer sep      = 1pt,
                                                     minimum size   = 8 pt,
                                                     scale=0.9}}
                      \tikzset{EdgeStyle/.style   = {draw, thin, scale=0.9}}
                      \tikzset{LabelStyle/.style =   {text = black, draw=none, scale=0.9}}
                      \tikzset{Label7Style/.style =   {text = black,outer sep = 1pt,dashed,thin, scale=0.9}}
\node[VertexStyle] {12}
	child
	{
		node[Label7Style] {126}
		child{
				node[Vertex7Style] {26}
				}
		child{
				node[Vertex7Style] {61}
				}
		}
	child{
		node[LabelStyle] {124}
		child{
				node[VertexStyle] {24}
				child{
						node[LabelStyle] {243}
						child{
								node[VertexStyle] {43}
								}
						child{
								node[VertexStyle] {32}
								}
						}
				}
		child{
				node[VertexStyle] {41}
				child{
						node[LabelStyle] {134}
						child{
								node[VertexStyle] {13}
								}
						child{
								node[VertexStyle] {34}
								child{
										node[Label7Style] {346}
										child{
												node[Vertex7Style] {46}
												}
										child{
												node[Vertex7Style] {63}
												}
										}
								child{
										node[Label7Style] {347}
										child{
												node[Vertex7Style] {47}
												}
										child{
												node[Vertex7Style] {73}
												}
										}
								}
						}
				child [missing]
				child [missing]
			}
		}
	child{
		node[LabelStyle] {123}
		child{
				node[VertexStyle] {23}
				child{
						node[LabelStyle] {235}
						child{
								node[VertexStyle] {35}
								}
						child{
								node[VertexStyle] {52}
								}
						}
				child
				{
						node[Label7Style] {236}
						child{
								node[Vertex7Style] {36}
								}
						child{
								node[Vertex7Style] {62}
								child{
														node[Label7Style] {276}
														child{
																node[Vertex7Style] {27}
																}
														child{
																node[Vertex7Style] {76}
																}
														}
								}
						}
				child
				{
						node[Label7Style] {237}
						child{
								node[Vertex7Style] {37}
								}
						child{
								node[Vertex7Style] {72}
								}
						}
				}
		child{
				node[VertexStyle] {31}
				child{
						node[LabelStyle] {153}
						child{
								node[VertexStyle] {15}
								}
						child{
								node[VertexStyle] {53}
								}
						}
			}
		}
	child{
		node[LabelStyle] {125}
		child{
				node[VertexStyle] {25}
				child{
						node[LabelStyle] {254}
						child{
								node[VertexStyle] {54}
								}
						child{
								node[VertexStyle] {42}
								}
						}
				}
		child{
				node[VertexStyle] {51}
				child{
						node[LabelStyle] {145}
						child{
								node[VertexStyle] {14}
								}
						child{
								node[VertexStyle] {45}
								child
								{
										node[Label7Style] {456}
										child{
												node[Vertex7Style] {56}
												}
										child{
												node[Vertex7Style] {64}
												}
										}
								child
								{
										node[Label7Style] {457}
										child{
												node[Vertex7Style] {57}
												}
										child{
												node[Vertex7Style] {74}
												}
										}
								}
						child [missing]
						}
				child
				{
						node[Label7Style] {165}
						child{
								node[Vertex7Style] {16}
								child{
										node[Label7Style] {167}
										child [missing]
										child{
											node[Vertex7Style] {67}
												}
										child{
											node[Vertex7Style] {71}
											}
										}
								}
						child{
								node[Vertex7Style] {65}
								}
						}
				child
				{
						node[Label7Style] {175}
						child{
								node[Vertex7Style] {17}
								}
						child{
								node[Vertex7Style] {75}
								}
						}
			}
		}
	;
\end{tikzpicture}
\vspace{-15pt}
\caption{The nearly spanning tree $T_7$ constructed from $T_5$} \label{fig:T7}
\end{figure*}
\end{example}

\section{A Recursive Construction}
\label{sec:recursive}

In this section we present the recursive construction for a
$(2n+1,M_{2n+1},\mK)$-snake from a
$(2n-1,M_{2n-1},\mK)$-snake. The construction is based on the
nearly spanning tree $T_{2n+1}$ presented in the previous section.
Each of its vertices represent a class in which a $\mK$-snake based
on the $(2n-1,M_{2n-1},\mK)$-snake is generated. Those $\mK$-snakes
are merged together into one $(2n+1,M_{2n+1},\mK)$-snake using the
framework presented in the previous section. We conclude this
section with an analysing the length of the generated
$\mK$-snake compared the total number of permutations in $S_{2n+1}$.

We generate a $(2n+1,M_{2n+1},\mK)$-snake, $C_{2n+1}$,
whose transitions sequence is
${t_{k_1} , t_{k_2} ,\ldots , t_{k_{M_{2n+1}}}}$.
$C_{2n+1}$ has the following properties:
\begin{enumerate}
\item[(Q1)] $k_j$ is odd for all ${j\in[M_{2n+1}]}$.
\item[(Q2)] $k_{M_{2n+1}}=2n+1$.
\item[(Q3)] For each $z\in [2n+1]$ there exists a permutation $\pi \in C_{2n+1}$ such that $\pi(2n+1)=z$.
\end{enumerate}
The starting point of the recursive construction is $2n+1=3$.
The transitions sequence for $2n+1=3$ is $t_3,t_3,t_3$,
and the complete $(3,3,\mK)$-snake is $C_3\eqdef \{[1, 2, 3], [3, 1, 2], [2, 3, 1]\}$.
Clearly (Q1), (Q2), and (Q3) hold for this transitions sequence and $C_3$.

Now, assume that there exists a $(2n-1,M_{2n-1},\mK)$-snake, $C_{2n-1}$,
which satisfies properties (Q1), (Q2), (Q3), and let
$\cT_{2n-1}=t_{k_1} , t_{k_2} , \ldots , t_{k_{M_{2n-1}}}$
be its S-skeleton, i.e., $\cT_{2n-1}$ is the transitions sequence of $C_{2n-1}$.
Note that (Q1), (Q2), and (Q3) depend on the transitions sequence $\cT_{2n-1}$
and are independent of the first permutation of $C_{2n-1}$.
We construct a ${(2n+1,M_{2n+1},\mK)}$-snake, $C_{2n+1}$,
where $M_{2n+1}=((2n+1)(2n)-1)M_{2n-1}$, which also satisfies (Q1), (Q2), and (Q3).

First, all the permutations of $A_{2n+1}$ are
partitioned into ${(2n+1)(2n)}$ classes according
to the last ordered two elements in the permutations.
This implies that (P1) and (P2) are satisfied.
In addition, (P3) and (P4) for $\cT_{2n-1}$ are immediately
implied by (Q2) and (Q3) for  $C_{2n-1}$, respectively.
Hence $\cT_{2n-1}$ can be used
as the S-skeleton for the $\mK$-snakes in $L_{2n+1}$.
Now, we merge the $\mK$-snakes of the classes in $L_{2n+1}$ (except $[2,1]$),
by using Lemma~\ref{prop:mergeClasses1} and the nearly spanning tree $T_{2n+1}$ of
Theorem~\ref{thm:treeCompose}.
We have to show that (Q1), (Q2), and (Q3) are satisfied for $C_{2n+1}$. (Q1) is readily verified.
Clearly, $t_{2n+1}$ was used to obtain $C_{2n+1}$ (see Lemma~\ref{prop:mergeClasses1}),
and therefore we can always define $\cT_{2n+1}$
in such a way that its last transition is $t_{2n+1}$, and hence (Q2) is satisfied.
For each $z\in [2n+1]$ there exists a class $[x,z]$
whose $\mK$-snake is joined into $C_{2n+1}$,
and therefore (Q3) is satisfied.
Thus, we have
\begin{theorem}
Given a $(2n-1,M_{2n-1},\mK)$-snake which satisfies (Q1), (Q2), and (Q3),
we can obtain a $(2n+1,M_{2n+1},\mK)$-snake,
where $M_{2n+1}=((2n+1)(2n)-1)M_{2n-1}$,
which also satisfies (Q1), (Q2), and (Q3).
\end{theorem}

Following~\cite{YeSc12},
we define $D_{2n+1}=\frac{M_{2n+1}}{(2n+1)!}$ as the ratio between
the number of permutations in the given ${(2n+1, M_{2n+1},\mK)}$-snake
and the size of $S_{2n+1}$.
Recall that if $C$ is an $(2n+1,M,\mK)$-snake then $M\leq \frac{{|S_{2n+1}|}}{2}$,
and we conjecture that optimal size is ${M=(2n+1)!-2n+1}$.
Thus, it is desirable to obtain a value $D_{2n+1}$ close to half as much as possible.
In our recursive construction $M_{2n+1}=((2n+1)(2n)-1)M_{2n-1}$.
Thus, we have
\begin{align*}
D_3&=\frac{1}{2},\\
\prod_{n=2}^{\infty}\frac{D_{2n+1}}{D_{2n-1}}&=
\frac{12\sqrt{\pi}}{5(1+\sqrt{5})\Gamma(\frac{1}{4}(5-\sqrt{5}))\Gamma(\frac{1}{4}(1+\sqrt{5}))},
\end{align*}
which implies that
\begin{align*}
\lim\limits_{n\to \infty} D_{2n+1}&= \frac{1}{2}\cdot \frac{12\sqrt{\pi}}{5(1+\sqrt{5})\Gamma(\frac{1}{4}(5-\sqrt{5}))\Gamma(\frac{1}{4}(1+\sqrt{5}))}\\
&\approx 0.4338.
\end{align*}
This computation can be done by any mathematical tool, e.g., WolframAlpha.
This improves on the construction described in \cite{YeSc12}, which yields $M_{2n+1}=(2n+1)(2n-1)M_{2n-1}$ and $\lim\limits_{n\to \infty} D_{2n+1}=\lim\limits_{n\to \infty} \frac{1}{\sqrt{\pi n}}$.

\section{A Direct Construction based on Necklaces}
\label{sec:direct}

In this section we describe a direct construction to
form a $(2n+1,M_{2n+1},\mK)$-snake. First, we describe
a method to partition the classes which were used before
into subclasses that are similar to necklaces. Next, we show how
subclasses from different classes are merged into
disjoint chains. Finally, we present a hypergraph and a graph
in which we have to search for certain trees to form our desired $\mK$-snake which
we believe is of maximum length. Such $\mK$-snakes were found in $S_7$ and $S_9$.

We present a direct construction for a $(2n+1,M_{2n+1},\mK)$-snake, $C_{2n+1}$.
The goal is to obtain $M_{2n+1}=\frac{(2n+1)!}{2}-(2n-1)$,
and hence ${\frac{D_{2n+1}}{D_{2n-1}}\geq 1-\frac{1}{(2n)!}}$.
We believe that there is always a $(2n+1,M_{2n+1},\mK)$-snake with
$M_{2n+1}=\frac{(2n+1)!}{2}-(2n-1)$ and there is no such $\mK$-snake
with more codewords.
We are making a slight change in the framework discussed in Section~\ref{sec:main}.
First, all the permutations of $A_{2n+1}$ are
partitioned into ${(2n+1)(2n)}$ classes according to the last ordered two elements.
We denote by $[x,y]$ the class of all even permutations
in which the last ordered pair in the permutation is $(x,y)$.
Each class is further partitioned into subclasses
according to the cyclic order of the first $2n-1$ elements in the permutations,
i.e., in each class $[x,y]$, the $\frac{(2n-1)!}{2}$ permutations
are partitioned into $\frac{(2n-2)!}{2}$ disjoint subclasses.
This implies that (P1) and (P2) are satisfied for both classes and subclasses.
Let's denote each one of the subclasses by $[\alpha]-[x,y]$ where $\alpha$
is the cyclic order of the first $2n-1$ elements in the permutations of the subclass.
Let $\alpha_1, \alpha_2$ be two permutations over $[2n+1]\setminus \{x,y\}$.
If $\alpha_1$ and $\alpha_2$ have the same cyclic order, we denote it by $\alpha_1 \simeq \alpha_2$,
otherwise $\alpha_1 \not\simeq \alpha_2$.
Note that if $\alpha_1 \simeq \alpha_2$ then $[\alpha_1]-[x,y]=[\alpha_2]-[x,y]$.
For example $[1,2,3]-[4,5]$ represents the subclass
with the permutations $[1,2,3,4,5]$, $[3,1,2,4,5]$, and $[2,3,1,4,5]$.

Let $L_{2n+1}$ be the set of all classes,
and let $\cT=t_{2n-1}^{2n-1}$ be the S-skeleton of the $\mK$-snakes in $L_{2n+1}$.
Note that a $\mK$-snake generated by $\cT$
spans exactly all the permutations in one subclass.
Hence (P3) and (P4) are immediately implied for both classes and subclasses.
Such a $\mK$-snake will be called a \emph{necklace}.
The slight change in the framework is that instead of one $\mK$-snake,
each class contains $\frac{(2n-2)!}{2}$ $\mK$-snakes,
all of them have the same S-skeleton.

The necklaces (subclasses) $[\alpha]-[x,y]$ are similar to
necklaces on $2n-1$ elements. Joining the necklaces into
one large $\mK$-snake might be similar to the join of cycles from
the pure cycling register of order $2n-1$, PCR$_{2n-1}$,
into one cycle, which is also known as a de Bruijn sequence~\cite{EtLe84,Fred82}.
There are two main differences between the two types of necklaces.
The first one is that in de Bruijn sequences
the necklaces do not represent permutations, but words
of a given length over some finite alphabet. The second is
that there is rather a simple mechanism to join all the
necklaces into a de Bruijn sequence.
We would like to have such a mechanism to join
as many as possible necklaces from all the classes into one $\mK$-snake.

Let $T_{2n+1}$ be the nearly spanning tree constructed by Theorem~\ref{thm:treeCompose}.
By repeated application of Lemma~\ref{prop:mergeClasses1}
according to the hyperedges of $T_{2n+1}$ starting from a necklace in the class $[1,2]$
we obtain a $\mK$-snake which contains exactly one necklace
from each class $[x,y]\ne [2,1]$.
Such a $\mK$-snake will be called a \emph{chain}.
If the chain contains the necklace $[\alpha]-[1,2]$, we will denote it by $c[\alpha]$.
For two permutations $\alpha_1$ and $\alpha_2$ over $[2n+1]\setminus \{1,2\}$
such that $\alpha_1 \simeq \alpha_2$ we have $c[\alpha_1]=c[\alpha_2]$.
Note that there is a unique way to merge the three necklaces which correspond to a hyperedge of $T_{2n+1}$,
and hence there is no ambiguity in $c[\alpha]$ (even so the order of the joins is not unique),
Note also that the transitions sequence of two distinct chains is usually different.
The number of permutations in a chain is $((2n+1)(2n)-1)(2n-1)$.
The following lemma is an immediate consequence of Lemma~\ref{prop:mergeClasses1}.

\begin{lemma}
\label{lem:key1}
Let $[x,y]$, $[y,z]$, and $[z,x]$ be three classes,
and let $\alpha$ be a permutation of $[2n+1]\setminus \{x,y,z\}$.
The necklaces $[\alpha,z]-[x,y]$, $[\alpha,y]-[z,x]$,
and $[\alpha,x]-[y,z]$ can be merged together, where
$\alpha,z$ is the sequence formed by concatenation of $\alpha$
and $z$.
\end{lemma}

\begin{lemma}
\label{lem:key2}
Let $[x,y]$, $[y,z]$, and $[z,x]$ be three classes.
All the subclasses in these classes can be partitioned into disjoint sets,
where each set contains exactly one necklace
from each of the above three classes.
The necklaces of each set can be merged together into one $\mK$-snake.
\end{lemma}
\begin{proof}
For each permutation $\alpha$ over $[2n+1]\setminus \{x,y,z\}$,
the necklaces $[\alpha,z]-[x,y]$, $[\alpha,y]-[z,x]$,
and $[\alpha,x]-[y,z]$ can be merged by Lemma \ref{lem:key1}.
Thus, all the subclasses in these classes can be partitioned into disjoint sets.
\end{proof}

\begin{cor}
\label{cor:disjointChains}
The permutations of all the classes except for $[2,1]$
can be partitioned into disjoint chains.
\end{cor}

By Corollary~\ref{cor:disjointChains} we construct $\frac{(2n-2)!}{2}$
disjoint chains which span $A_{2n+1}$, except for all the even permutations of the class $[2,1]$.
Recall that we have the same number, $\frac{(2n-2)!}{2}$, of $[2,1]$-necklaces, which
span all the permutations of the class~$[2,1]$.
Now, we need a method to merge all these chains and necklaces, except
for one necklace from the class~$[2,1]$, into one $\mK$-snake $C_{2n+1}$.
Note that for $2n+1=5$ we have only one chain. Thus, this chain is the final $\mK$-snake $C_5$.
This $\mK$-snake is exactly the same $\mK$-snake as the one generated
by the recursive construction in Section~\ref{sec:recursive}.

\begin{lemma}
\label{lem:chainConnection}
Let $x$ be an integer such that $3\leq x\leq 2n+1$,
let $\alpha$ be a permutation of $[2n+1]\setminus \{x,2,1\}$,
and assume that the permutations $[\alpha,1,x,2]$ and $[\alpha,2,1,x]$ are contained in two distinct chains.
We can merge these two chains via the necklace $[\alpha,x]-[2,1]$.
\end{lemma}
\begin{proof}
Let $c_1$ be the chain which contains the permutation $\pi_1=[\alpha,1,x,2]$,
$c_2$ be the chain which contains the permutation $\pi_2=[\alpha,2,1,x]$, and
$\eta$ be the necklace which contains the permutation $\pi_3=[\alpha,x,2,1]$.
Note that all the chains contains only the p-transitions $t_{2n+1}$ and $t_{2n-1}$.
The permutation $t_{2n+1}(\pi_1)$ appears in $c_2$,
the permutation $t_{2n+1}(\pi_2)$ appears in $\eta$,
and the permutation $t_{2n+1}(\pi_3)$ appears in $c_1$.
Therefore, $\pi_1$, $\pi_2$, and $\pi_3$ are followed
by $t_{2n-1}$ in $c_1$, $c_2$, and $\eta$, respectively.
Let $\sigma_i$, $i\in \{1,2\}$, be a transitions sequence such that $\sigma_i, t_{2n-1}$ is the transitions sequence of $c_i$,
and therefore $t_{2n-1}(\sigma_i(\pi_i))=\pi_i$.
By Lemma \ref{lem:helpLem1} we have $\sigma_1 \leftrightsquigarrow t_{2n-1}^{2n-2} \leftrightsquigarrow \sigma_2$.
Similarly to Lemma~\ref{prop:mergeClasses1},
by replacing the transition $t_{2n-1}$ which follows $\pi_3$ in $\eta$,
with $t_{2n+1} \circ \sigma_{1} \circ t_{2n+1} \circ \sigma_2 \circ t_{2n+1}$,
we merge $c_1$, $c_2$ and $\eta$ into a $\mK$-snake.
Thus, we have
\begin{align*}
\pi_3=&[a_1, a_2, \ldots,  a_{2n-2},x,2,1] \\
&\downarrow {t_{2n+1}}  \\
&[ 1, a_1, a_2, \ldots, a_{2n-2}, x, 2] \\
&\downarrow {\sigma_1 \leftrightsquigarrow t_{2n-1}^{2n-2}} \hphantom{aaaaaaaaaaaa}  the\ chain\ c_1 \\
\pi_1=&[ a_1, a_2, \ldots, a_{2n-2}, 1,x,2] \\
&\downarrow {t_{2n+1}}  \\
&[ 2, a_1, a_2, \ldots, a_{2n-2},1,x] \\
&\downarrow {\sigma_2 \leftrightsquigarrow t_{2n-1}^{2n-2}} \hphantom{aaaaaaaaaaaa}  the\ chain\ c_2 \\
\pi_2=&[ a_1, a_2, \ldots, a_{2n-2}, 2, 1, x] \\
&\downarrow {t_{2n+1}} \hphantom{aaaaaaa} return\ to\ the\ necklace\ \eta \\
t_{2n-1}(\pi_3) =&[x, a_1, a_2, \ldots, a_{2n-2}, 2,1]
\end{align*}
\end{proof}

For each $x$, $3\leq x\leq 2n+1$, and for each permutation $\alpha$ of $[2n+1]\setminus\{x,1,2\}$,
the merging of two distinct chains which contain the permutations
$[\alpha,1,x,2]$ and $[\alpha,2,1,x]$ via the necklace $[\alpha,x]-[2,1]$
as described in Lemma~\ref{lem:chainConnection}, will be denoted by $M[x]$-connection.
Note that if $x\in\{3,4,5\}$ then the permutations $[\alpha,1,x,2]$ and $[\alpha,2,1,x]$
are contained in the same chain.
Thus, there are no $M[3]$-connections, $M[4]$-connections, or $M[5]$-connections.

Lemma~\ref{lem:chainConnection} suggests a method
to join all the chains and all the $[2,1]$-necklaces except one into
a $\mK$-snake of length $\frac{(2n+1)!}{2}-(2n-1)$.
This should be implemented by $\frac{(2n-2)!}{2}-1$
iterations of the merging suggested by Lemma~\ref{lem:chainConnection}.
The current merging problem is also translated into a $3-graph$ problem (see Definition~\ref{def:3-graph}).
Let $\hat{H}_{2n+1}=(\hat{V}_{2n+1},\hat{E}_{2n+1})$ be a 3-graph defined as follows.
\begin{align*}
\hat{V}_{2n+1}&=\{c[\alpha]~:~\alpha \text{ is a permutation of } [2n+1]\setminus \{1,2\}\}\\
& \cup \{[\beta]-[2,1]~:\\
& \hphantom{\cup\{\}} \beta \text{ is a permutation of } [2n+1]\setminus \{1,2\}\} \\
\hat{E}_{2n+1}&=\{\{c[\alpha_1], c[\alpha_2], [\beta]-[2,1]\}~:\\
& \hphantom{=\{\}} c[\alpha_1] \text{ and } c[\alpha_2] \text{ can be merged together} \\
& \hphantom{=\{\}} \text{via } [\beta]-[2,1] \text{ by Lemma \ref{lem:chainConnection}}\}.
\end{align*}
The vertices in $\hat{V}_{2n+1}$ are of two types, chains and $[2,1]$-necklaces.
Each $e\in \hat{E}_{2n+1}$ contains three vertices, two chains and one necklace,
which can be merged together by Lemma \ref{lem:chainConnection}.
Therefore, the edge will be signed by $M[x]$ as described before.
Note that $\hat{E}_{2n+1}$ might contains parallel edges with different signs.

Let $\hat{T}_{2n+1}=(V_{\hat{T}_{2n+1}},E_{\hat{T}_{2n+1}})$ be a nearly spanning tree in $\hat{H}_{2n+1}$.
Note that such a nearly spanning tree must contain all the vertices in $\hat{V}_{2n+1}$
except for one $[2,1]$-necklace.
If such a nearly spanning tree exists then
by Lemma~\ref{lem:chainConnection},
we can merge all the chains via $[2,1]$-necklaces
to form the $\mK$-snake $C_{2n+1}$.
This $\mK$-snake contains all the permutations
of $A_{2n+1}$ except for $2n-1$ permutations
which form one $[2,1]$-necklace.

The joins which are performed are determined by the edges of $\hat{T}_{2n+1}$.
Note that there is a unique way to merge the three vertices
which correspond to a hyperedge of $\hat{T}_{2n+1}$ signed by $M[x]$.
Hence, by using the given spanning trees $T_{2n+1}$ and $\hat{T}_{2n+1}$,
there is no ambiguity in $C_{2n+1}$ (even so the orders of the joins are not unique).
However, different nearly spanning trees
can yield different final $\mK$-snakes.
Note that the $\mK$-snake $C_{2n+1}$ generated by this construction
has only $t_{2n+1}$ and $t_{2n-1}$ p-transitions, where usually $t_{2n-1}$ is used.
The p-transition $t_{2n-1}$ is the only transition in the $\mK$-snake of the subclasses.
On average $3$ out of $4n-2$ sequential p-transitions of $C_{2n+1}$ are the p-transition $t_{2n+1}$.
A similar property exists when a de Bruijn sequence is generated
from the necklaces of pure cycling register of order~$n$~\cite{EtLe84,Fred82}.

Finding a nearly spanning tree $\hat{T}_{2n+1}$ is an open question.
But, we found such trees for $n=3$ and $n=4$.
We believe that a similar construction to the
one which follows in the sequel for $n=3$ and $n=4$, exists for all $n>4$.

\begin{conjecture}
\label{conj:consB}
For each $n\geq 2$, there exists
a $(2n+1,M_{2n+1},\mK)$-snake,
where $M_{2n+1}=\frac{(2n+1)!}{2}-(2n-1)$
in which there are only $t_{2n-1}$ and $t_{2n+1}$ p-transitions.
\end{conjecture}

\begin{example}
\label{ex:consB_7}
For $n=3$, a $(7,2515,\mK)$-snake is constructed by using the tree $T_7$ of Example \ref{Exm:T7},
and the tree $\hat{T}_7$ defined below.
$\hat{T}_7$ contains $12$ chains, where each chain contains $41$ necklaces.
It also contains $11$ $[2,1]$-necklaces and $11$ hyperedges.
Denote an edge in $\hat{H}_{7}$ by $(\{ c_i, c_j, \eta_k \}, x)$
where $M[x]$ is the sign of the edge.
$\hat{T}_{7}$ is defined as follows.
\begin{tabular}{ll}
\underline{The chains in $\hat{T}_{7}$:} & \tabularnewline
$c_1\ =[3, 4, 5, 6, 7]-[1, 2]$, &
$c_2\ =[3, 4, 6, 7, 5]-[1, 2]$, \tabularnewline
$c_3\ =[3, 4, 7, 5, 6]-[1, 2]$, &
$c_4\ =[3, 5, 4, 7, 6]-[1, 2]$, \tabularnewline
$c_5\ =[3, 5, 6, 4, 7]-[1, 2]$, &
$c_6\ =[3, 5, 7, 6, 4]-[1, 2]$, \tabularnewline
$c_7\ =[3, 6, 4, 5, 7]-[1, 2]$, &
$c_8\ =[3, 6, 5, 7, 4]-[1, 2]$, \tabularnewline
$c_9\ =[3, 6, 7, 4, 5]-[1, 2]$, &
$c_{10}=[3, 7, 4, 6, 5]-[1, 2]$,\tabularnewline
$c_{11}=[3, 7, 5, 4, 6]-[1, 2]$,&
$c_{12}=[3, 7, 6, 5, 4]-[1, 2]$.
\end{tabular}
\begin{tabular}{ll}
\underline{The necklaces in $\hat{T}_{7}$:} & \tabularnewline
$\eta_1\ =[3, 4, 5, 7, 6]-[2, 1]$, &
$\eta_2\ =[3, 4, 6, 5, 7]-[2, 1]$, \tabularnewline
$\eta_3\ =[3, 4, 7, 6, 5]-[2, 1]$, &
$\eta_4\ =[3, 5, 4, 6, 7]-[2, 1]$, \tabularnewline
$\eta_5\ =[3, 5, 6, 7, 4]-[2, 1]$, &
$\eta_6\ =[3, 5, 7, 4, 6]-[2, 1]$, \tabularnewline
$\eta_7\ =[3, 6, 4, 7, 5]-[2, 1]$, &
$\eta_8\ =[3, 6, 5, 4, 7]-[2, 1]$, \tabularnewline
$\eta_9\ =[3, 6, 7, 5, 4]-[2, 1]$, &
$\eta_{10}=[3, 7, 4, 5, 6]-[2, 1]$, \tabularnewline
$\eta_{11}=[3, 7, 5, 6, 4]-[2, 1]$.
\end{tabular}
\begin{tabular}{ll}
\underline{The edges in $\hat{T}_{7}$:} & \tabularnewline
$e_1\ =(\{ c_{11}, c_6\hphantom{\scriptsize{i}}, \eta_9\hphantom{\scriptsize{i}} \}, 6)$, &
$e_2\ =(\{ c_6\hphantom{\scriptsize{i}}, c_1\hphantom{\scriptsize{i}}, \eta_2\hphantom{\scriptsize{i}} \}, 6)$, \tabularnewline
$e_3\ =(\{ c_2\hphantom{\scriptsize{i}}, c_{12}, \eta_{11} \}, 6)$, &
$e_4\ =(\{ c_{12}, c_7\hphantom{\scriptsize{i}}, \eta_4\hphantom{\scriptsize{i}} \}, 6)$, \tabularnewline
$e_5\ =(\{ c_5\hphantom{\scriptsize{i}}, c_3\hphantom{\scriptsize{j}}, \eta_3\hphantom{\scriptsize{i}} \}, 6)$, &
$e_6\ =(\{ c_3\hphantom{\scriptsize{i}}, c_4\hphantom{\scriptsize{i}}, \eta_7\hphantom{\scriptsize{i}} \}, 6)$, \tabularnewline
$e_7\ =(\{ c_9\hphantom{\scriptsize{i}}, c_{10}, \eta_{10} \}, 6)$, &
$e_8\ =(\{ c_{10}, c_8\hphantom{\scriptsize{i}}, \eta_5\hphantom{\scriptsize{i}} \}, 6)$, \tabularnewline
$e_9\ =(\{ c_{12}, c_9\hphantom{\scriptsize{i}}, \eta_8\hphantom{\scriptsize{i}} \}, 7)$, &
$e_{10}=(\{ c_9\hphantom{\scriptsize{i}}, c_3\hphantom{\scriptsize{i}}, \eta_1\hphantom{\scriptsize{i}} \}, 7)$, \tabularnewline
$e_{11}=(\{ c_2\hphantom{\scriptsize{i}}, c_{11}, \eta_6\hphantom{\scriptsize{i}} \}, 7)$.
\end{tabular}

$\hat{H}_{7}$ contains another $[2,1]$-necklace, $\eta_{12}=[3, 7, 6, 4, 5]-[2, 1]$,
and the following additional edges:
\begin{align*}
\begin{tabular}{ll}
$e_{12}=(\{ c_1\hphantom{\scriptsize{'}}, c_{11}, \eta_{12} \}, 6),$
&
$e_{13}=(\{ c_7\hphantom{\scriptsize{i}}, c_2\hphantom{\scriptsize{i}}, \eta_1\hphantom{\scriptsize{i}} \}, 6)$, \tabularnewline
$e_{14}=(\{ c_4\hphantom{\scriptsize{i}}, c_5\hphantom{\scriptsize{1}}, \eta_8\hphantom{\textbf{\scriptsize{i}}} \}, 6)$,  &
$e_{15}=(\{ c_8\hphantom{\scriptsize{i}}, c_9\hphantom{\scriptsize{i}}, \eta_6\hphantom{\scriptsize{i}} \}, 6)$, \tabularnewline
$e_{16}=(\{ c_{10}, c_2\hphantom{\scriptsize{i}}, \eta_2\hphantom{\scriptsize{i}} \}, 7)$,  &
$e_{17}=(\{ c_8\hphantom{\scriptsize{i}}, c_1\hphantom{\scriptsize{i}}, \eta_3\hphantom{\scriptsize{i}} \}, 7)$, \tabularnewline
$e_{18}=(\{ c_{11}, c_{10}, \eta_4\hphantom{\scriptsize{'}} \}, 7)$,  &
$e_{19}=(\{ c_3\hphantom{\scriptsize{i}}, c_{12}, \eta_5\hphantom{\scriptsize{i}} \}, 7)$,  \tabularnewline
$e_{20}=(\{ c_6\hphantom{\scriptsize{i}}, c_7\hphantom{\scriptsize{i}}, \eta_7\hphantom{\scriptsize{i}} \}, 7)$, &
$e_{21}=(\{ c_4\hphantom{\scriptsize{i}}, c_8\hphantom{\scriptsize{i}}, \eta_9\hphantom{\scriptsize{i}} \}, 7)$, \tabularnewline
$e_{22}=(\{ c_1\hphantom{\scriptsize{i}}, c_4\hphantom{\scriptsize{i}}, \eta_{10} \}, 7)$,  &
$e_{23}=(\{ c_5\hphantom{\scriptsize{i}}, c_6\hphantom{\scriptsize{i}}, \eta_{11} \}, 7)$, \tabularnewline
$e_{24}=(\{ c_7\hphantom{\scriptsize{i}}, c_5\hphantom{\scriptsize{i}}, \eta_{12} \}, 7)$.
\end{tabular}
\end{align*}
An additional different illustration of $\hat{H}_{7}$ is presented in the sequel (see Example~\ref{ex:g_7}).
\end{example}

For each $n\geq 3$,
let $\cG_{2n+1}=(\cV_{2n+1},\cE_{2n+1})$ be a multi-graph (with parallel edges) with labels and signs on the edges.
The vertices of $\cV_{2n+1}$ represent the $\frac{(2n-2)!}{2}$ chains
and hence $|\cV_{2n+1}|=\frac{(2n-2)!}{2}$.
There is an edge signed with $M[x]$, where $6\leq x\leq 2n+1$,
between the vertex (chain) $c_1$ and vertex (chain) $c_2$,
if $c_1$ contains a permutation $[\alpha,2,1,x]$
and $c_2$ contains the permutation $[\alpha,1,x,2]$, where $c_1\ne c_2$.
The label on this edge is the necklace $[\alpha,x]-[2,1]$.
Note that the label on the edge is a necklace
which can merge together the chains of its corresponding endpoints by $M[x]$-connection.
Note also that the pair $\alpha, x$ might not be unique
and hence the graph might have parallel edges.
A tree in $\cG_{2n+1}$ which doesn't have two edges with the same label,
will be called \emph{a chain tree}.
The following Lemma can be easily verified.
\begin{lemma}
There exists a nearly spanning tree in $\hat{H}_{2n+1}$
if and only if there exists a chain tree in $\cG_{2n+1}$.
\end{lemma}

Henceforth, $T_{2n+1}$ will be the nearly spanning tree constructed in Theorem \ref{thm:treeCompose},
and the chains are constructed via~$T_{2n+1}$.
\begin{definition}
Let $\cG_1=(\cV_1, \cE_1)$ and $\cG_2=(\cV_2,\cE_2)$ be two multi-graphs with labels and signs on the edges,
where the set of the labels of $\cG_i$ denoted by $\cL_i$, $i\in \{1,2\}$.
We say that
$\cG_1$ is isomorphic to $\cG_2$ if there exist two bijective functions $f: \cV_1\to \cV_2$ and $g: \cL_1 \to \cL_2$,
with the following property:
$(u,v)\in \cE_1$ with the label $\eta$ and sign $M[x]$, if and only if
$(f(u),f(u)) \in \cE_2$ with the label $g(\eta)$ and sign $M[x]$.
\end{definition}
\begin{definition}
For each $n\geq 4$, a sub-graph of $\cG_{2n+1}$ which is isomorphic to $\cG_{2n-1}$
is called a component of $\cG_{2n+1}$,
and denoted by $A=(\cV_A,\cL_A)$
where $\cV_A$ consists of the vertices (chains) of the component,
$\cL_A$ consists of the labels ($[2,1]$-necklaces) on the edges in the component.
Note that $|\cV_A|=|\cL_A|$, i.e., the numbers of the distinct labels is equal to the number of the vertices.
\end{definition}
\begin{definition}
Two components, $A=(\cV_A,\cL_A)$ and $B=(\cV_B,\cL_B)$, in $\cG_{2n+1}$ are called disjoint
if $\cV_A\cap \cV_B =\varnothing$ and $\cL_A\cap \cL_B =\varnothing$, i.e., there is no a common vertex (chain) or a common label ($[2,1]$-necklace) in $A$ and $B$.
\end{definition}
\begin{lemma}
\label{lem:2n+1structure}
For each $n\geq 4$, $\cG_{2n+1}$ consists of $(2n-3)(2n-2)$ disjoint copies
of isomorphic graphs to $\cG_{2n-1}$, called components.
The edges between the vertices of two distinct components are signed
only with $M[2n]$ and $M[2n+1]$.
\end{lemma}
\begin{proof}
The $M[x]$-connections are
deduced by the tree $T_{2n+1}$, which was used for the construction of the chains.
In particular, the path between the vertices $[1,x]$ and $[x,2]$
in $T_{2n+1}$ determines the $M[x]$-connections in~$\cG_{2n+1}$.
By Theorem~\ref{thm:treeCompose}, $T_{2n-1}$ is a sub-graph of $T_{2n+1}$.
Therefore, for each $x$, $x\geq 3$,
the path between the vertices $[1,x]$ and $[x,2]$ in $T_{2n+1}$ is equal to
the path between the vertices $[1,x]$ and $[x,2]$ in $T_{2k+1}$ for each ${x\leq 2k+1\leq 2n+1}$.
The number of the vertices (chains) in $\cG_{2n+1}$ is equal to $\frac{(2n-2)!}{2}$,
and each component contains $\frac{(2n-4)!}{2}$ vertices.
Thus, $\cG_{2n+1}$ consists of $(2n-3)(2n-2)$ disjoint copies
of isomorphic graphs to $\cG_{2n-1}$ connected by edges signed only with $M[2n]$ and $M[2n+1]$.
\end{proof}

For each $n\geq 4$, let $\hat{\cG}_{2n+1}=(\hat{\cV}_{2n+1},\hat{\cE}_{2n+1})$ be the \emph{component graph} of
$\cG_{2n+1}$.
The vertices of $\hat{\cV}_{2n+1}$ represent the components of $\cG_{2n+1}$,
There is an edge signed with $M[x]$, $x\in \{2n,2n+1\}$,
between the vertices (components) $A$ and $B$,
if the chain that contains the permutation $[\alpha,2,1,x]$ is contained in $A$,
and the chain that contains the permutation $[\alpha,1,x,2]$ is contained in $B$.
The label on this edge is the necklace $[\alpha,x]-[2,1]$.
We define $\hat{\cG}_{7}$ to be $\cG_{7}$, i.e., each component of $\hat{\cG}_{7}$ consists of exactly one chain
(and also one distinct $[2,1]$-necklace in order to follow
the properties of $\hat{\cG}_{2n+1}$).
\begin{definition}
A components spanning tree, $\hat{T}_{2n+1}$
is a spanning tree in $\hat{\cG}_{2n+1}$,
where in the set of the labels of the tree's edges,
there are no two labels from the same component, i.e.,
each label in the set of the labels of the tree's edges belongs to a different component.
\end{definition}
\begin{example}
\label{ex:g_7}
$\hat{\cG}_{7}$ is depicted in Figure~\ref{fig:hatG7}, where the vertices numbers and the edges labels
corresponds to the chains and the necklaces in Example \ref{ex:consB_7}, respectively.
The vertical edges are signed with $M[6]$,
while the horizontal edges are signed with $M[7]$.
The double lines edges correspond to the edges of $\hat{T}_7$.
\begin{figure}[H]
\centering
\vspace{-10pt}
\begin{tikzpicture} [-,scale=0.6,node distance=3cm,
  thick,main node/.style={font={\normalsize}, circle,draw,minimum size=25 pt,scale=0.6}
  ,edge_node/.style={font={\normalsize},scale=0.6}
  ,edgeTree_node/.style={font={\large},scale=0.6}]

  \node[main node] (1) {1};
  \node[main node] (11) [below of=1] {11};
  \node[main node] (6) [below of=11] {6};

  \node[main node] (2) [right of=1] {2};
  \node[main node] (12) [below of=2] {12};
  \node[main node] (7) [below of=12] {7};

  \node[main node] (3) [right of=2] {3};
  \node[main node] (4) [below of=3] {4};
  \node[main node] (5) [below of=4] {5};

  \node[main node] (8) [right of=3] {8};
  \node[main node] (9) [below of=8] {9};
  \node[main node] (10) [below of=9] {10};

  \path[]
    (1) edge node[edge_node] [left, xshift=.1cm] {12} (11)

    (7) edge [bend left] node[edge_node] [left,xshift=.1cm] {1} (2)

    (4) edge node[edge_node] [left,xshift=.1cm] {8} (5)

    (8) edge node[edge_node] [right,xshift=-.1cm] {6} (9);

  \path[]
    (11) edge[double, thin] node[edgeTree_node] [left,xshift=.1cm] {9} (6)
    (6) edge [double, thin,bend left] node[edgeTree_node] [left,xshift=.1cm] {2} (1)

    (2) edge[double,thin] node[edgeTree_node] [left,xshift=.1cm, yshift=-.6cm] {11} (12)
    (12) edge[double,thin] node[edgeTree_node] [left,xshift=.1cm] {4} (7)

    (5) edge [double,thin,bend left] node[edgeTree_node] [left,xshift=.1cm] {3} (3)
    (3) edge[double,thin] node[edgeTree_node] [left,xshift=.1cm, yshift=-.1cm] {7} (4)

    (9) edge[double,thin] node[edgeTree_node] [right,xshift=-.1cm] {10} (10)
    (10) edge [double,thin,bend left] node[edgeTree_node] [right,xshift=-.1cm] {5} (8)

    (2) edge[double,thin] node[edge_node] [above] {6} (11)

    (12) edge[double,thin,bend right] node[edge_node] [above, xshift=.7cm] {8} (9)
    (9) edge[double,thin] node[edge_node] [above, xshift=-0.8cm, yshift=0.8cm] {1} (3);

  \path[]
    (6) edge node[edge_node] [above,yshift=-.1cm] {7} (7)
    (7) edge node[edge_node] [above,yshift=-.1cm] {12} (5)
    (5) edge [bend left] node[edge_node] [above,yshift=-.1cm] {11} (6)

    (1) edge node[edge_node] [above, xshift=-2cm, yshift=1cm] {10} (4)
    (4) edge node[edge_node] [right, xshift=0.5cm, yshift=0.9cm] {9} (8)
    (8) edge[bend right=20] node[edge_node] [above] {3} (1)

    (11) edge node[edge_node] [above] {4} (10)
    (10) edge [bend right] node[edge_node] [above,xshift=-.4cm,yshift=.4cm] {2} (2)

    (3) edge node[edge_node] [above] {5} (12);

`\end{tikzpicture}
\vspace{-5pt}
\caption{The graph $\hat{\cG}_7$ and its component spanning tree $\hat{T}_7$} \label{fig:hatG7}
\end{figure}
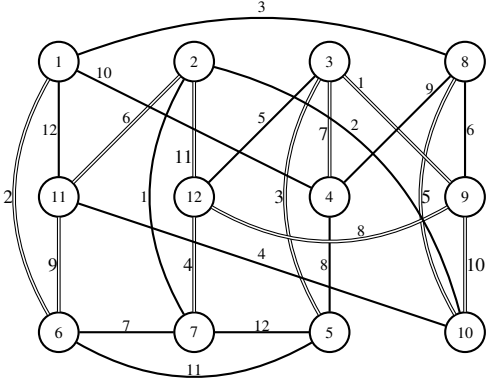
\end{example}

\begin{conjecture}
\label{conj:main}
For each component $A$ in $\hat{\cG}_{2n+1}$, $n\geq 3$,
and for each label $\eta$ of $A$,
there exists a components spanning tree,
where there is no edge in the tree with the label $\eta$.
\end{conjecture}
Conjecture \ref{conj:main} implies Conjecture \ref{conj:consB}, i.e.,
\begin{theorem}
If Conjecture \ref{conj:main} is true then for each $n\geq 2$, there exists
a $(2n+1,M_{2n+1},\mK)$-snake,
where $M_{2n+1}=\frac{(2n+1)!}{2}-(2n-1)$
in which there are only $t_{2n-1}$ and $t_{2n+1}$ p-transitions.
\end{theorem}
Conjecture \ref{conj:main} was verified by computer search for $n=3$ and $n=4$.
By using Conjecture \ref{conj:main} recursively,
for each $n\geq 3$, and for each necklace $\eta$ in class $[2,1]$,
we can construct a chain tree $T$ in $\cG_{2n+1}$,
which doesn't include $\eta$ as a label on an edge in $T$.

\begin{cor}
There exist a $(7,2515,\mK)$-snake and a~$(9,181433,\mK)$-snake,
and hence $\lim\limits_{n\to \infty} \frac{M_{2n+1}}{S_{2n+1}}\approx 0.4743$.
\end{cor}

Note that the ratio $\lim\limits_{n\to \infty} \frac{M_{2n+1}}{S_{2n+1}}$ would be improved,
if there exists a $(2m+1,\frac{(2m+1)!}{2}-(2m-1),\mK)$-snake for some $m>4$.

\begin{conjecture}
The $(2n-3)(2n-2)$ components in $\hat{\cG}_{2n+1}$
can be arranged in a $(2n-3)\times (2n-2)$ grid.
The edges which are sign by $M[2n]$ define $2n-2$ cycles of length $2n-3$.
Each cycle contains the vertices of exactly one column, and is called an $M[2n]$-cycle.
The edges which are sign with $M[2n+1]$ are between two components in different columns,
and they also define $2n-2$ cycles of length $2n-3$.
Such a cycle will be called an $M[2n+1]$-cycle.
Each multi-edge between two components has $\frac{(2n-4)!}{2}$ parallel edges (the number of chains in the component).
Parallel edges have the same sign $x$, $x\in \{2n,2n+1\}$,
but different labels (i.e., $M[x]$-connection, but with different $[2,1]$-necklaces).
\end{conjecture}
\begin{example}
An illustration for the structure of $\hat{\cG}_{2n+1}$ for $n=3$ is presented in Example~\ref{ex:g_7},
and for $n=4$ is depicted in Figure~\ref{fig:hatG9}.
In $\hat{\cG}_9$ there are $30$ components,
where each component is isomorphic to $\hat{\cG}_7$
(thus, it contains $12$ chains and $12$ $[2,1]$-necklaces).
\begin{figure}[H]
\vspace{-5pt}
\begin{tikzpicture} [-, scale=0.5, node distance=1.5 cm,
  thick,main node/.style={font={\normalsize},circle,draw,scale=0.5,minimum size=20 pt}
  ,edge node/.style={font={\normalsize},scale=0.5}]

  \node[main node] (6) {1};		
  \node[main node] (7) [below of=6]{2};		
  \node[main node] (8) [below of=7]{3};		
  \node[main node] (9) [below of=8]{4};		
  \node[main node] (10) [below of=9]{5};		

  \node[main node] (1)[right= 1.15cm of 6] {6};					
  \node[main node] (2) [below of=1]{7};	 	
  \node[main node] (3) [below of=2]{8};			
  \node[main node] (4) [below of=3]{9};			
  \node[main node] (5) [below of=4]{10};			

  \node[main node] (26) [right= 1.15cm of 1]{11};		
  \node[main node] (27) [below of=26]{12};		
  \node[main node] (28) [below of=27]{13};		
  \node[main node] (29) [below of=28]{14};		
  \node[main node] (30) [below of=29]{15};		

  \node[main node] (11) [right= 1.15cm of 26]{16};		
  \node[main node] (12) [below of=11]{17};		
  \node[main node] (13) [below of=12]{18};		
  \node[main node] (14) [below of=13]{19};		
  \node[main node] (15) [below of=14]{20};		

  \node[main node] (16) [right= 1.15cm of 11]{21};		
  \node[main node] (17) [below of=16]{22};		
  \node[main node] (18) [below of=17]{23};		
  \node[main node] (19) [below of=18]{24};		
  \node[main node] (20) [below of=19]{25};		

  \node[main node] (21) [right= 1.15cm of 16]{26};		
  \node[main node] (22) [below of=21]{27};		
  \node[main node] (23) [below of=22]{28};		
  \node[main node] (24) [below of=23]{29};		
  \node[main node] (25) [below of=24]{30};		

  \path[thin]
    (1) edge (2)
    (2) edge (3)
    (3) edge (4)
    (4) edge (5)
    (5) edge [bend left=20] (1)

    (6) edge (7)
    (7) edge (8)
    (8) edge (9)
    (9) edge (10)
    (10) edge [bend left=20] (6)

    (11) edge (12)
    (12) edge (13)
    (13) edge (14)
    (14) edge (15)
    (15) edge [bend left=20] (11)

    (16) edge (17)
    (17) edge (18)
    (18) edge (19)
    (19) edge (20)
    (20) edge [bend left=20] (16)

    (21) edge (22)
    (22) edge (23)
    (23) edge (24)
    (24) edge (25)
    (25) edge [bend left=20] (21)

    (26) edge (27)
    (27) edge (28)
    (28) edge (29)
    (29) edge (30)
    (30) edge [bend left=20] (26);

\path[thin]
    (25) edge (28)
    (28) edge[bend left=16] (10)
    (10) edge (3)
    (3)  edge (12)
    (12) edge (25)

    (15) edge[bend right=16] (23)
    (23) edge[bend left=12, out=60, in=200] (5)
    (5)  edge (9)
    (9)  edge[bend left=16] (19)
    (19) edge (15)

    (1)  edge (26)
    (26) edge (11)
    (11) edge (16)
    (16) edge (21)
    (21) edge[bend right=16] (1)

    (7)  edge (2)
    (2)  edge (29)
    (29) edge[bend left=16] (17)
    (17) edge (13)
    (13) edge (7)

    (22) edge[bend right=3, out=10, in=200] (8)
    (8)  edge (14)
    (14) edge (30)
    (30) edge[bend right=16] (18)
    (18) edge (22)

    (24) edge[out=125, in=305] (6)
    (6)  edge (27)
    (27) edge (4)
    (4)  edge (20)
    (20) edge (24);
\end{tikzpicture}
\vspace{-30pt}
\caption{The graph $\hat{\cG}_9$} \label{fig:hatG9}
\end{figure}
\end{example}

\section{Conclusions and Future Research}
\label{sec:conclude}

Gray codes for permutations using the operation push-to-the-top
and the Kendall's $\tau$-metric were discussed.
We have presented a framework for constructing snake-in-the-box codes
for $S_n$. The framework for the construction yield a recursive construction
with large snakes. A direct construction to obtain snakes which might
be optimal in length was also presented. Several questions arise from our discussion
and they are considered for current and future research.
\begin{enumerate}
\item Complete the direct construction for snakes of length
$\frac{(2n+1)!}{2} - 2n+1$ in $S_{2n+1}$.

\item Can a snake in $S_{2n+1}$ have size larger than
$\frac{(2n+1)!}{2} - 2n+1$?

\item Prove or disprove that the length of the longest snake in $S_{2n}$
is not longer than the length of the longest snake in~$S_{2n-1}$.

\item Examine the questions in this paper for the $\ell_\infty$ metric.
\end{enumerate}

\begin{center}
{\bf Acknowledgment}
\end{center}
The authors would like to thank the anonymous reviewers for their careful reading of the paper.
Especially, one of the reviewers pointed out on the good ratio $\lim\limits_{n\to \infty} \frac{M_{2n+1}}{S_{2n+1}}\approx 0.4338$ compared to one in \cite{YeSc12}.

\vspace{1cm}

\textbf{Michal Horovitz} was born in Israel in 1987. She received the
B.Sc. degree from the Open University of Israel, Ra'anana, Israel, in 2009,
from the department of Mathematics and from the department of Computer Science.
She is currently a Ph.D. student in the Computer Science Department at the Technion
- Israel Institute of Technology, Haifa, Israel.
Her research interests include coding theory with applications to non-volatile memories
and combinatorics.

\vspace{1cm}

\textbf{Tuvi Etzion} (M'89--SM'94--F'04) was born in Tel Aviv, Israel,
in 1956. He received the B.A., M.Sc., and D.Sc. degrees from the
Technion - Israel Institute of Technology, Haifa, Israel, in 1980,
1982, and 1984, respectively.

From 1984 he held a position in the Department of Computer Science
at the Technion, where he has a Professor position. During the
years 1985-1987 he was Visiting Research Professor with the
Department of Electrical Engineering - Systems at the University
of Southern California, Los Angeles. During the summers of 1990
and 1991 he was visiting Bellcore in Morristown, New Jersey.
During the years 1994-1996 he was a Visiting Research Fellow in
the Computer Science Department at Royal Holloway College, Egham,
England. He also had several visits to the Coordinated Science
Laboratory at University of Illinois in Urbana-Champaign during
the years 1995-1998, two visits to HP Bristol during the summers
of 1996, 2000,  a few visits to the Department of Electrical
Engineering, University of California at San Diego during the
years 2000-2012, and several visits to the Mathematics Department
at Royal Holloway College, Egham, England, during the years
2007-2009.

His research interests include applications of discrete
mathematics to problems in computer science and information
theory, coding theory, and combinatorial designs.

Dr Etzion was an Associate Editor for Coding Theory for the IEEE
Transactions on Information Theory from 2006 till 2009.
From 2004 to 2009, he was an Editor for the
Journal of Combinatorial Designs.
From 2011 he is an Editor for Designs, Codes, and Cryptography.
From 2013 he is an Editor for Advances of Mathematics in Communications.

\vspace{1cm}

\end{document}